\documentclass{article}
\usepackage{spconf,amsmath,graphicx, amssymb, hyperref, amsthm, dsfont, amsfonts}
\usepackage{algorithm,algorithmic}
\usepackage{multirow}
\usepackage{multicol}

\newcommand{\arbitraryScalar}{a}
\newcommand{\arbitraryMat}{\mathbf{A}}
\newcommand{\arbitraryVec}{\mathbf{a}}

\newcommand{\mat}[1]{\mathbf{#1}}
\renewcommand{\vec}[1]{\mathbf{#1}}

\newcommand*{\of}[1]{_{#1}}
\newcommand*{\tran}{^{\mathsf{T}}}
\newcommand*{\herm}{^{\mathsf{H}}}
\renewcommand*{\th}{^\text{th}}

\newcommand*{\inv}{^{-1}}
\newcommand{\norm}[1]{\lVert {#1} \rVert}
\newcommand{\abs}[1]{\left|{#1}\right|}
\newcommand{\normalized}[1]{{#1}\of{\text{norm}}}
\renewcommand{\max}[1]{{\text{max}}\left({#1}\right)}
\newcommand{\hadamard}{\odot}
\newcommand{\permute}{\pi}
\newcommand{\conj}[1]{\bar{#1}}

\newcommand{\generalVariationMetric}{\mu}

\newcommand{\totalVariation}{\sigma}
\newcommand{\totalVariationOf}[1]{\totalVariation(#1)}

\newcommand{\eigval}[1]{\lambda(#1)}
\newcommand{\eigvalof}[2]{\lambda\of{#2}({#1})}
\newcommand{\eigvalvec}[1]{\boldsymbol{\lambda}(#1)}
\newcommand{\eigvalmat}{\mat{\Lambda}}
\newcommand{\eigvec}{\vec{x}}
\newcommand{\eigmat}{\mat{X}}
\newcommand{\eigvalmax}{\max{\abs{\eigvalvec{\adjmat}}}}

\newcommand{\graph}[1]{{\cal G}({#1})}
\newcommand{\adjmat}{\mat{A}}
\newcommand{\gsignal}{s}
\newcommand{\gsignalVec}{\vec{\gsignal}}

\newcommand{\edge}[2]{#1\rightarrow #2}

\newcommand{\IGFT}[1]{\mathbf{GFT}\inv\of{#1}}
\newcommand{\GFT}[1]{\mathbf{GFT}\of{#1}}

\newcommand{\set}[1]{\mathbb {#1}}
\newcommand{\smat}[2]{\mathcal{M}\of{#1}(\set{#2})}
\newcommand{\suni}[2]{\mathcal{U}\of{#1}(\set{#2})}
\newcommand{\sdiag}[2]{\mathcal{D}\of{#1}(\set{#2})}
\newcommand{\spsd}[2]{\mathcal{S}\of{#1}(\set{#2})}

\newcommand{\polarUni}{\mat{Q}}
\newcommand{\polarPsdLeft}{\mat{P}}
\newcommand{\polarPsdRight}{\mat{F}}

\newcommand{\svdLeft}{\mat{U}}
\newcommand{\svdRight}{\mat{V}}
\newcommand{\svdDiag}{\mat{\Sigma}}
\newcommand{\svd}{\svdLeft\svdDiag\svdRight\tran}

\newcommand{\angularFreq}{\omega}

\newcommand{\cociteMat}{\mat{C}\of{\text{out}}}
\newcommand{\bibMat}{\mat{B}\of{\text{in}}}

\newcommand{\gaussian}[2]{\mathcal{N}(#1,#2)}

\newcommand{\equivalent}{\cong}

\makeatletter
\renewcommand*\env@matrix[1][\arraystretch]{%
  \edef\arraystretch{#1}%
  \hskip -\arraycolsep
  \let\@ifnextchar\new@ifnextchar
  \array{*\c@MaxMatrixCols c}}
\makeatother

\newtheorem{definition}{Definition}
\newtheorem{theorem}{Theorem}
\newtheorem{remark}{Remark}[definition]

\title{Frequency analysis and filter design for directed graphs\\with polar decomposition}

\name{Semin Kwak$^*$, Laura Shimabukuro$^*$, and Antonio Ortega}
\address{University of Southern California, Los Angeles, USA}

\begin{document}
\ninept
\maketitle
\def\thefootnote{*}\footnotetext{These authors contributed equally to this work.}\def\thefootnote{\arabic{footnote}}

\begin{abstract}
In this study, we challenge the traditional approach of frequency analysis on directed graphs, which typically relies on a single measure of signal variation such as total variation. We argue that the inherent directionality in directed graphs necessitates a multifaceted analytical approach that incorporates multiple signal variations definitions. Our methodology leverages the polar decomposition to define two distinct variations, each associated with different matrices derived from this decomposition. This approach provides a novel interpretation in the node domain and reveals aspects of graph signals that may be overlooked with a singular measure of variation.
Additionally, we develop graph filters specifically designed to smooth graph signals in accordance with our proposed variations. These filters allow for bypassing costly filtering operations associated with the original graph through effective cascading. We demonstrate the efficacy of our methodology using an M-block cyclic graph example, validating our claims and showcasing the advantages of our multifaceted approach in analyzing signals on directed graphs.
\let\thefootnote\relax\footnotetext{All figures presented in this manuscript are reproducible using the code available at \url{https://github.com/semink/ICASSP2024}.}
\end{abstract}


\section{Introduction}
Graph signal processing (GSP) has gained significant attention by effectively integrating traditional signal processing techniques with graph theory, enabling data processing on irregular domains~\cite{ortega8347162}. 
Many GSP results are well-developed for undirected graphs, leveraging the symmetry of edge connections, which ensures that graph operators (also known as shift operators) of undirected graphs, such as the adjacency matrix and the graph Laplacian, are diagonalizable, leading to the definition of the Graph Fourier Transform (GFT)~\cite{Ortega2020GSP}.

On the other hand, despite a growing need for applications that involve directed graphs 
(digraphs)~\cite{sarcheshmehpour2021federated, giraldo2020semi}, 
the intrinsic asymmetry of their edge connections presents challenges in defining GFTs. First, the basic graph operators become non-diagonalizable for many digraphs \cite{barrufet2020graph,barrufet2021orthogonal}. 
Second, ``ordering" the graph Fourier basis to interpret them as low or high frequencies is difficult because the eigenvalues, which can be used for the ordering, can be complex-valued. As a consequence, a variety of digraph GFTs have been proposed to address these two challenges. 

For example, studies directly using the adjacency matrix as the graph operator have suggested utilizing the Jordan decomposition~\cite{Sandryhaila_2013_2, Domingos_2020} or the singular value decomposition~\cite{Chen_2022} of the adjacency matrix for constructing GFT.
Alternatively, a recent study adds a minimal number of edges so that the corresponding graph operator on the augmented graph becomes diagonalizable~\cite{9325908}. 
Symmetrization of the adjacency matrix has also been studied~\cite{symmetrizations, li2023eigen}.
Finally, other studies define the graph operator as decomposed matrices of adjacency matrix, producing diagonalizable matrices~\cite{MHASKAR2018611, unitaryshift}. 

Although differing in their definition of GFTs, these previous studies coincide in using a single graph operator to measure signal variation, which is the key property in ordering the graph Fourier basis. 
Instead, we propose that analyzing the smoothness of signals on digraphs using only a single variation can be unnecessarily restrictive. Specifically, we suggest considering both (i) signal variation between disconnected nodes (which we call \textit{indirect variation}) and (ii) variation between connected nodes (which we call \textit{in-flow variation}) on a graph that has been minimally augmented to permit circular flows. 
These two variations offer complementary insights about signals on digraphs. 
%To define these multiple types of variation, 
We utilize the polar decomposition of an adjacency matrix, which results in two different terms:  1) a unitary matrix and 2) a positive semidefinite matrix, and show that the signal variations based on those matrices are related to the in-flow and indirect variations, respectively. 

There have been prior attempts to provide a spectral representation of each term of the polar decomposition~\cite{MHASKAR2018611, unitaryshift}. However, these approaches examined only one term from the polar decomposition and did not develop an interpretation of the resulting frequencies in the node domain. 
Here, we connect the polar decomposition with the in-flow and indirect variations, providing a coherent node-domain interpretation.
A significant advancement presented in this study compared to our preliminary pre-print work~\cite{shimabukuro2023signal} is the formulation of graph filters grounded in polar decomposition.
We demonstrate that the proposed graph filters, based on polynomials of graph operators derived from the polar decomposition, provide an alternative to those based on polynomials of the adjacency matrix. Importantly, this design negates the need for spectral domain filtering and the associated computationally expensive eigendecomposition.
With M-block cyclic graphs~\cite{mblock}, we elucidate the applicability of these filters on graph signals and demonstrate that analyzing graph signals by leveraging polar decomposition, considering both in-flow and indirect variations, provides richer information than analyses based solely on a single variation.

The rest of the paper is structured as follows. In Section~\ref{sec:preliminary}, we provide key notations and our definition of the GFT.
%used in this study. 
In Section~\ref{sec:polar_decomposition}, we introduce polar decomposition in the context of digraphs. In Section~\ref{sec:node_domain_interpretation}, we analyze graph signal variation via polar decomposition, provide a precise node-domain interpretation, and filter designs based on the polar decomposition. The effectiveness of our proposed approach is subsequently validated in Section~\ref{sec:case_study}.
%through detailed case studies. 
Finally, Section~\ref{sec:conclusions} offers a conclusion and summarizes our findings.

\section{Preliminaries}
\label{sec:preliminary}
\subsection{Notation}

The set of square matrices of size \( n \) over the field \( \set{F} \) (\( \set{R} \) or \( \set{C} \)) is denoted as \( \smat{n}{\set{F}} \).
The sets of unitary matrices, positive semi-definite matrices, and square diagonal matrices of size \( n \) over \( \set{F} \) are respectively represented as \( \suni{n}{\set{F}} \), \( \spsd{n}{\set{F}} \), and \( \sdiag{n}{\set{F}} \). 
Bold uppercase and lowercase letters denote matrices and vectors, respectively. The \( i\th \) column vector of a matrix \( \arbitraryMat \) is indicated as \( \arbitraryVec\of{i} \).
The \( i\th \) entry of a vector \( \arbitraryVec \) is \( \arbitraryScalar\of{i} \).
The operations \( \arbitraryMat\tran \), \( \conj{\arbitraryMat} \), \( \arbitraryMat\herm \), \(\Re(\arbitraryMat)\) and \(\Im(\arbitraryMat)\) stand for the transpose, conjugate, conjugate-transpose, real component, and imaginary component of the matrix \( \arbitraryMat \), respectively, which also hold for vector and scalar.
We use \( \abs{\arbitraryVec} \) as the element-wise absolute value of the vector \( \arbitraryVec \). The element-wise product of two vectors or matrices is denoted as \( \hadamard \).
The function \( \max{\cdot} \) of a vector (or matrix) returns the value that is the maximum value among all entries of the vector (or matrix).
Eigenvalues of the matrix \( \arbitraryMat \) are  represented in vector form as \( \eigvalvec{\arbitraryMat} \), with the \( i\th \) eigenvalue given by \( \eigvalof{\arbitraryMat}{i} \).

A graph \( \graph{V,E} \) is defined by a set of nodes \( V \) and edges \( E \), with an edge \( e_{ij} \) representing a connection from node \( j \) to node \( i \). In a weighted graph, the edges \( e_{ij} \) have real weights \( a_{ij} \).
The adjacency matrix \( \adjmat\in\smat{n}{R} \) of the graph \( \graph{V,E} \) is a matrix that has the \( i\th \) row and \( j\th \) column entry as \( a_{ij} \), corresponding to the edge weight from node \( j \) to node \( i \).
% We denote G(A) as a directed graph on n vertices with an edge ij if and only if the ij-th entry of A is non-zero.
We say the graph \( \graph{\adjmat} \) as \textit{supports} \( \adjmat \) when its edge weights are represented by the adjacency matrix \( \adjmat \).
A graph signal is represented by a real vector \( \gsignalVec\in\set{R}^n \), wherein the real scalar entry \( \gsignal\of{i} \) is the signal associated with node \( i \).

\subsection{Graph Fourier transform (GFT)}

%\par\noindent{\textbf{Graph Fourier transform (GFT)}}
\begin{definition}[Total variation on 
graphs~\cite{Sandryhaila_2013_2}]
    Let \( \gsignalVec \) be a graph signal on a graph \( \graph{\adjmat} \). The total variation of \( \gsignalVec \) is defined as
    \[
    \totalVariationOf{\gsignalVec}=\norm{\normalized{\adjmat}\gsignalVec-\gsignalVec}\of{1},
    \]
    where \( \normalized{\adjmat}=\frac{1}{\eigvalmax}\adjmat \) and a graph signal \( \gsignalVec\of{1} \) is said to have more energy in the high-frequency components compared to \( \gsignalVec\of{2} \) on \( \graph{\adjmat} \) with respect to total variation if \( \totalVariationOf{\gsignalVec\of{1}}>\totalVariationOf{\gsignalVec\of{2}} \).
\end{definition}

% \begin{remark}
% Comparing graph signals based on their variation isn't solely confined to total variation. Thus, one could state that a graph signal \( \gsignalVec\of{1} \) has more high-frequency components than \( \gsignalVec\of{2} \) on \( \graph{\adjmat} \) relative to a variation metric \( \generalVariationMetric \) if \( \generalVariation{\gsignalVec\of{1}}>\generalVariation{\gsignalVec\of{2}} \).
% \end{remark}

\begin{definition}[Graph Fourier transform on directed graphs~\cite{Sandryhaila_2013_2}]
    Consider a graph \( \graph{\adjmat} \) with diagonalizable adjacency matrix: 
    %\( \adjmat \)  
    \[
    \adjmat=\eigmat\eigvalmat\eigmat^{-1}.
    \]
    The set of eigenvectors (columns of \( \eigmat \)) is termed as the graph Fourier basis of \( \graph{\adjmat} \). 
    When the basis is ordered such that for \( i<j \), the eigenvector \( \eigvec_i \) exhibits lower variation than or equal to that of \( \eigvec_j \) for a given variation measure \( \generalVariationMetric \), the ordered matrix \( \eigmat \) is denoted as \( \IGFT{\adjmat|\generalVariationMetric} \) and its inverse \( \eigmat\inv \) as \( \GFT{\adjmat|\generalVariationMetric} \).
    % Furthermore, the graph frequency spectra of a graph signal \( \gsignalVec \) are given by \( \gft{\adjmat|\generalVariationMetric}{\gsignalVec}=\GFT{\adjmat|\generalVariationMetric}\gsignalVec \).
\end{definition}

% \begin{remark}
%     Throughout this paper, the term ``eigendecomposition" means that with 1-norm normalized eigenvectors.
% \end{remark}

\begin{remark}\label{remark:conj}
    When \( \adjmat\in\smat{n}{R} \), if \( \adjmat\eigvec=\lambda\eigvec \) then \( \adjmat\conj{\eigvec}=\conj{\lambda}\conj{\eigvec} \).
\end{remark}

\begin{remark}
    For simplicity, throughout this paper, we focus on diagonalizable adjacency matrices \( \adjmat \) to discuss the graph Fourier transform (GFT), but our proposed method applies to non-diagonalizable \(\adjmat\). %For a broader perspective, refer to~\cite{Sandryhaila_2013_2}. 
\end{remark}

\begin{remark}
    Also, for brevity, we omit the sub-index \( \generalVariationMetric \) (variation measure function) when it is the total variation, i.e., if \( \generalVariationMetric = \totalVariation\) then we write \( \GFT{\adjmat|\totalVariation}= \GFT{\adjmat} \).
    %and \( \gft{\adjmat}{\gsignalVec}=\gft{\adjmat|\totalVariation}{\gsignalVec} \).
\end{remark}

\begin{remark}\label{remark:tv_bound}
    The \( i\th \) \(\ell_1\)-norm normalized eigenvector \( \eigvec\of{i} \) of \( \adjmat \) is such that \( \totalVariationOf{\eigvec\of{i}}=\abs{1-\frac{\eigvalof{\adjmat}{i}}{\eigvalmax}} \), which takes values between 0 and 2 (between 0 and 1 when \(\adjmat\) is positive semidefinite).
    %, it is bounded between 0 and 1.
    % This not only allows for ordering the eigenvectors solely by eigenvalues but also associating them into real numbers ranging from 0 to 2, and 
\end{remark}

\begin{definition}[Equivalence of graphs w.r.t.~GFT]
    Graphs \( \graph{\adjmat\of{1}} \) and \( \graph{\adjmat\of{2}} \) are said to be  \textbf{equivalent} for a given GFT over \( \generalVariationMetric \) if both have the same ordered set of normalized eigenvectors with the variation metric \( \generalVariationMetric \), i.e., \( \GFT{\adjmat\of{1}|\generalVariationMetric}=\GFT{\adjmat\of{2}|\generalVariationMetric} \).
    Equivalence is denoted as \( \graph{\adjmat\of{1}}\overset{\generalVariationMetric}{\equivalent} \graph{\adjmat\of{2}} \).
\end{definition}

\begin{remark}
    For brevity, we omit \( \generalVariationMetric \) when the measure is the total variation, i.e., \( \graph{\adjmat\of{1}}{\equivalent} \graph{\adjmat\of{2}} \) means \( \graph{\adjmat\of{1}}\overset{\totalVariation}{\equivalent} \graph{\adjmat\of{2}} \).
\end{remark}

% \begin{remark}
% Despite potential disparities in the edge connectivity between \( \graph{\adjmat\of{1}} \) and \( \graph{\adjmat\of{2}} \), a graph signal \( \gsignalVec \) will exhibit identical frequency components across these  graphs, as long as \( \graph{\adjmat\of{1}}\equivalent\graph{\adjmat\of{2}} \). Therefore, when intricate edge connectivity in \( \graph{\adjmat\of{1}} \) complicates the understanding of signal variation, a clearer interpretation can be achieved using \( \graph{\adjmat\of{2}} \) if its connectivity is simpler.    
% \end{remark}

\begin{figure}[t]
	\centering
	\includegraphics[width=.95\linewidth]{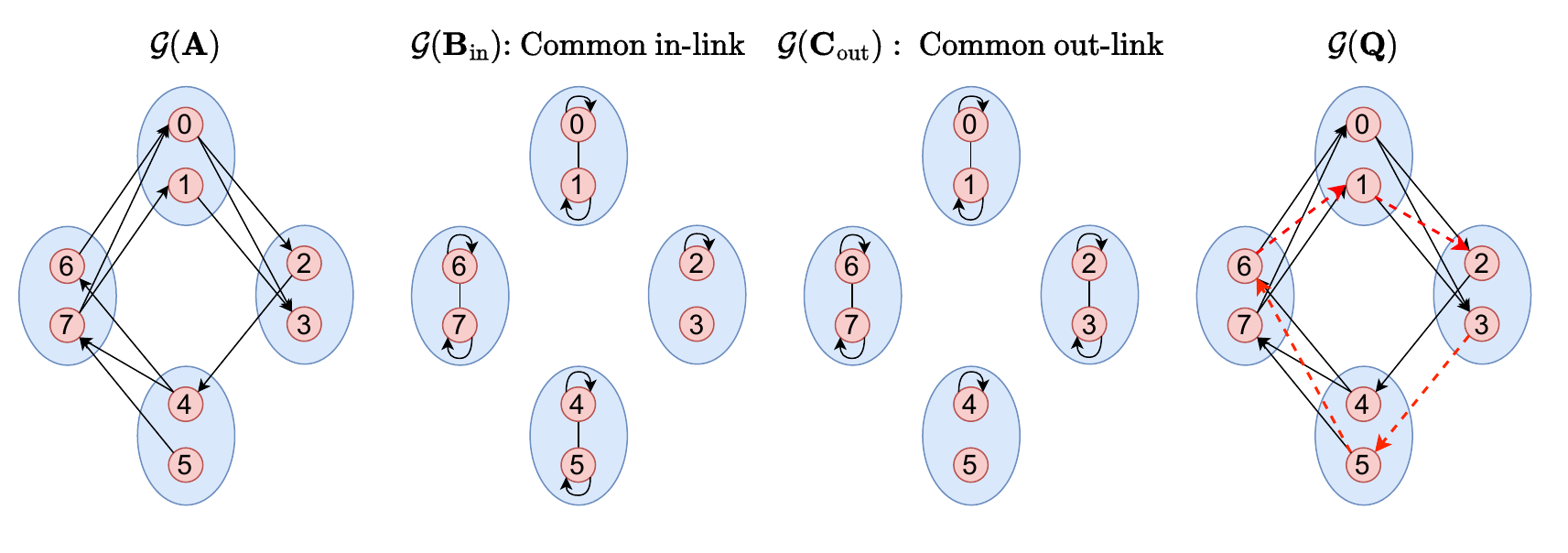}
	\caption{Example of an original digraph \( \graph{\adjmat} \) and 
    symmetrizations \( \graph{\bibMat} \) and \( \graph{\cociteMat} \) where \( \bibMat=\adjmat\adjmat\tran \) and \( \cociteMat = \adjmat\tran\adjmat \), which can be viewed as graphs whose edges connect nodes that share (i) common in-links and (ii) common out-links on the original graph, respectively. 
    The graph that supports the unitary matrix in the polar decomposition (\( \graph{\polarUni} \)) includes added edges (red dashed arrows) that result in Hamiltonian cycles.  \( \graph{\polarUni}\) has real edge weights.}
	\label{fig:symm_examples}
\end{figure}

\section{Polar decomposition}
\label{sec:polar_decomposition}
We propose to use the polar decomposition to analyze graph signals on directed graphs more effectively. 
%, therefore, this subsection introduces the foundational concepts of polar decomposition.
% Introduce that we are going to leverage Polar decomposition of analyzing graph signals, so we introduce basics of the polar decomposition.
The polar decomposition of a matrix is analogous to the polar form of a complex value (\( z=re^{j\theta} \), where \( r\ge 0 \) and \( j=\sqrt{-1} \)). It decomposes a matrix through the matrix multiplication of a unitary matrix (analogous to \( e^{j\theta} \)) 
and a positive semi-definite matrix (analogous to \( r \)): 
%as illustrated by
\begin{equation}
\label{eq:polar-decomp}
\adjmat = \polarUni \polarPsdLeft = \polarPsdRight \polarUni,
\end{equation}
where \( \adjmat\in\smat{n}{R} \), \( \polarUni\in\suni{n}{R} \), \( \polarPsdLeft\in\spsd{n}{R} \), and \( \polarPsdRight\in\spsd{n}{R} \).

The polar decomposition can be obtained from the singular value decomposition. 
For a given matrix \( \adjmat\in\smat{n}{R} \), the singular value decomposition is represented as \( \adjmat=\svdLeft\svdDiag\svdRight\tran \), where \( \svdLeft\in\suni{n}{R} \), \( \svdDiag\in\sdiag{n}{R}\cap\spsd{n}{R} \), and \( \svdRight\in\suni{n}{R} \).
Consequently, we obtain the terms in \eqref{eq:polar-decomp} as \( \polarUni=\svdLeft\svdRight\tran \), \( \polarPsdLeft=\svdRight\svdDiag\svdRight\tran \), and \( \polarPsdRight=\svdLeft\svdDiag\svdLeft\tran \)~\cite{shimabukuro2023signal}.

Given that \( \polarUni \) is a real-valued unitary matrix, all eigenvalues of \( \polarUni \) lie on the unit circle and come in conjugate pairs with associated conjugate eigenvectors.
% (refer to Remark~\ref{remark:conj}).
Consequently, eigenvectors can be ordered by the phase of associated eigenvalues utilizing the variation measure \( \generalVariationMetric\of{\angularFreq}({\eigval{\polarUni})}=\angularFreq\of{\polarUni} \) for the eigenvalue \( \eigval{\polarUni}=e^{j\angularFreq\of{\polarUni}} \) (\( -\pi\le\angularFreq\of{\polarUni}\le\pi \)), therefore, the eigenvector can be associated to the corresponding graph (angular) frequency \( \angularFreq\of{\polarUni} \).
As a result, the graph Fourier transform of \( \graph{\polarUni} \) is defined as \( \GFT{\polarUni|\generalVariationMetric\of{\angularFreq}} \). 
%Note that the ordering is consistent with the total variation, i.e., \(
%\GFT{\polarUni|\generalVariationMetric\of{\angularFreq}}=\GFT{\polarUni}
%\).
% The conjugate pair nature allows interpreting only positive frequencies without loss of generality \cite{unitaryshift}.

In contrast, eigenvalues of \( \polarPsdLeft \) and \( \polarPsdRight \) are non-negative real values due to the positive semi-definite nature of these matrices. As a result, the total variation of normalized eigenvectors for these matrices is confined within the interval \([0,1]\) (refer to Remark~\ref{remark:tv_bound}).
The graph Fourier transforms of \( \graph{\polarPsdLeft} \) and \( \graph{\polarPsdRight} \) are defined as \( \GFT{\polarPsdLeft} \) and \( \GFT{\polarPsdRight} \), respectively.
\section{Node-domain interpretation}
\label{sec:node_domain_interpretation}

In this section, we clarify how the respective graph signal variations observed in  \( \graph{\polarPsdLeft} \), \( \graph{\polarPsdRight} \) and \( \graph{\polarUni} \) relate to the signal variation of the primary graph \( \graph{\adjmat} \) from a node-domain perspective. 

The graphs \( \graph{\polarPsdLeft} \) and \( \graph{\polarPsdRight} \) are directly related to the symmetrizations of the directed graph \( \graph{\adjmat} \). 
With \( \bibMat = \adjmat\adjmat\tran \) and \( \cociteMat = \adjmat\tran\adjmat \),  \( \graph{\bibMat} \) can be viewed as a graph with edges connecting nodes that share common in-links on the original graph \( \graph{\adjmat} \), while  \( \graph{\cociteMat} \) can be viewed as a graph with edges connecting nodes that share common out-links (see Fig.~\ref{fig:symm_examples})~\cite{symmetrizations}. 
Since \( \adjmat=\svd \), their eigendecompositions are: 
\begin{equation}
\bibMat =  \svdLeft\svdDiag^2\svdLeft\tran\text{ and }\cociteMat = \svdRight\svdDiag^2\svdRight\tran.
\end{equation}
Noting that \( \polarPsdLeft = \svdRight \svdDiag \svdRight^{\tran} \) and \( \polarPsdRight = \svdLeft \svdDiag \svdLeft^{\tran} \), we can deduce that \( \graph{\polarPsdLeft} \) is equivalent to \( \graph{\cociteMat} \) and \( \graph{\polarPsdRight} \) is equivalent to \( \graph{\bibMat} \). This is because they share the same eigenvectors, corresponding to \(\svdLeft\) and \(\svdRight\), where the eigenvectors are ordered from low to high frequency in an order that is the same for \( \svdDiag \) and \( \svdDiag^2 \).
Therefore, given \( \graph{\polarPsdLeft} \equivalent \graph{\cociteMat} \), signal variations on \( \graph{\polarPsdLeft} \) reflect the variations between nodes \( i \) and \( j \) in \( \graph{\adjmat} \) that share a descendant node \( k \) such that both edges \( \edge{i}{k} \) and \( \edge{j}{k} \) exist on the graph \( \graph{\adjmat} \). 
Similarly, the signal variations on \( \graph{\polarPsdRight} \) represent variations between nodes \( i \) and \( j \) in \( \graph{\adjmat} \) that share an ancestor node \( k \), given that both edges \( \edge{k}{i} \) and \( \edge{k}{j} \) exist on the graph \( \graph{\adjmat} \). 
In light of these observations, we term the signal variations on \( \graph{\polarPsdLeft} \) and \( \graph{\polarPsdRight} \) as \textbf{indirect variations} on \( \graph{\adjmat} \).

On the other hand, we denote the signal variation on \( \graph{\polarUni} \) as \textbf{in-flow variation} on \( \graph{\adjmat} \). A fundamental characteristic of \( \polarUni \) is its proximity to \( \adjmat \): it is the unitary matrix closest to \( \adjmat \) with respect to any unitarily invariant norm~\cite{graph_mining}. 
This relationship suggests that \( \polarUni \) is the matrix (graph operator) nearest to \( \adjmat \) such that the signal norm is preserved after multiplication\footnote{The operation \(\polarUni\gsignalVec\) can be interpreted as ``shift" in the signal processing sense. The signal \(\gsignalVec\) can be perfectly reconstructed by the ``reverse shift" \(\polarUni\tran\).}, that is \( \norm{\gsignalVec}_2 = \norm{\polarUni\gsignalVec}_2 \).
Furthermore, in cases where the graph \( \graph\polarUni \) is strongly connected, it is conjectured to be Hamiltonian~\cite{severini2003digraph, gutin2006hamilton}, implying that graph has a cycle containing each vertex exactly once. 
In situations where \( \graph{\polarUni} \) becomes disjoint, its individual sub-graphs retain the property of being supportive of unitary matrices, thereby ensuring that a Hamiltonian cycle persists within every sub-graph. 
This suggests that every node belongs to a minimum of one cycle (including a potential self-loop) within \( \graph{\polarUni} \).
Given this understanding, \( \graph{\polarUni} \) can be viewed as an alternative representation of \( \graph{\adjmat} \) that introduces additional edges to facilitate cycles. 
For instance, as depicted in Fig.~\ref{fig:symm_examples}, the graph \( \graph{\polarUni} \) incorporates four extra edges (illustrated as dashed red arrows), which allow the introduction of Hamiltonian cycles (e.g., 0-2-4-6-1-3-5-7-0).
Consequently, the signal variation on \( \graph{\polarUni} \) captures not only the variation between nodes \( i \) and \( j \) connected by the edge \( \edge{i}{j} \) in \( \graph{\adjmat} \), but also the variation between the disconnected nodes \( j \) and \( k \) in \( \graph{\adjmat} \), provided their linkage potentially forms a cycle. 
In subsequent sections, we will empirically demonstrate the advantages of analyzing graph signals via \( \GFT{\polarUni} \) rather than \( \GFT{\adjmat} \), given its dual properties of close similarity to \( \adjmat \) and its Hamiltonian cycle inclusion.

\begin{figure}[t]
	\centering
	\includegraphics[width=0.69\linewidth]{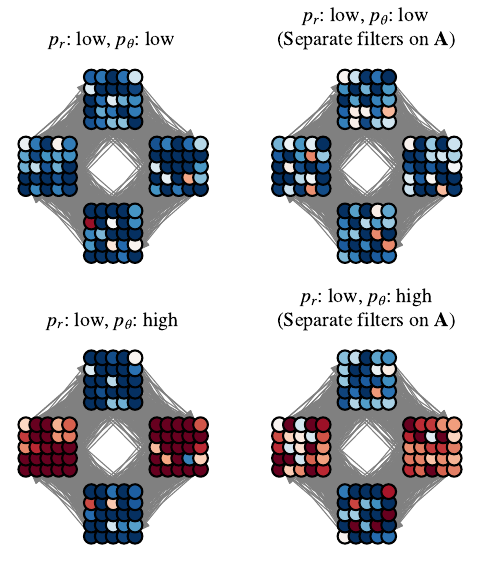}
	\caption{
 A randomly generated graph signal is filtered by the cascaded filter (\(\vec{y} =p_\theta(\polarUni)p_r(\polarPsdLeft)\gsignalVec\)) and the separate filters on \(\adjmat\) (\(\vec{y} =\eigmat p_\theta(\mat{\Theta})p_r(\mat{\Sigma})\eigmat^{-1}\gsignalVec\)) of different polynomial filters \(p_r(\polarPsdLeft)\) and \( p_\theta(\polarUni) \). The colors red, white, and blue represent negative, zero, and positive values. The intensity of the color indicates the magnitude.
 The lowpass filter \(p_r\) smooths signals within the same block while \(p_\theta\) filters block-to-block signals. 
 Similarity is observed between the filtered signals based on the cascaded filter \(p_\theta(\polarUni)p_r(\polarPsdLeft)\) and the filter \( \eigmat p_\theta(\mat{\Theta})p_r(\mat{\Sigma})\eigmat^{-1} \) although the adjacency matrix is not normal.
 }
	\label{fig:filter_output}
    \end{figure}

\subsection{Normal adjacency matrix case}
If \( \adjmat \) is normal, i.e., \( \adjmat\adjmat\tran = \adjmat\tran\adjmat \), a direct relationship can be identified between the magnitude and phase components of its eigenvalues given by \( \eigvalvec{\adjmat} = \abs{\eigvalvec{\adjmat}}\hadamard e^{j\angle{\eigvalvec{\adjmat}}} \), and the eigenvalues of \( \polarPsdLeft \), \( \polarPsdRight \), and \( \polarUni \). Specifically, when \( \adjmat \) is normal, \( \polarPsdLeft=\polarPsdRight \). The magnitude component of each eigenvalue, \( \abs{\eigvalvec{\adjmat}} \), corresponds to an eigenvalue in \( \eigvalvec{\polarPsdLeft} \), implying \( \abs{\eigvalvec{\adjmat}}=\permute\of{1}(\eigvalvec{\polarPsdLeft}) \) for some permutation \( \permute\of{1} \). Similarly, the phase of each eigenvalue, \( e^{j\angle{\eigvalvec{\adjmat}}} \), aligns with an eigenvalue in \( \eigvalvec{\polarUni} \), so \( e^{j\angle{\eigvalvec{\adjmat}}}=\permute\of{2}(\eigvalvec{\polarUni}) \) for another permutation \( \permute\of{2} \).

By splitting \( \eigvalvec{\adjmat} \) into its magnitude and phase, we gain insights into the indirect variation and in-flow variation on \( \graph{\adjmat} \). However, focusing solely on the eigenvalues in \( \eigvalvec{\adjmat} \) from a total variation perspective may lead to misinterpretations: Eigenvectors associated with eigenvalues that have minimal differences in total variation might be mistakenly perceived as exhibiting similar variations. 
For instance, normalized eigenvectors associated with \(\lambda_1/\max{|\boldsymbol{\lambda}|}=e^{j\frac{\pi}{2}}\) and \(\lambda_2/\max{|\boldsymbol{\lambda}|}=(\sqrt{2}-1)e^{j\pi}\) yield the same total variation \(\sqrt{2}\) (as seen in Remark~\ref{remark:tv_bound}), even though they reflect entirely different indirect and in-flow variations.

\begin{figure*}[t]
		\centering
        \includegraphics[width=.95\linewidth]{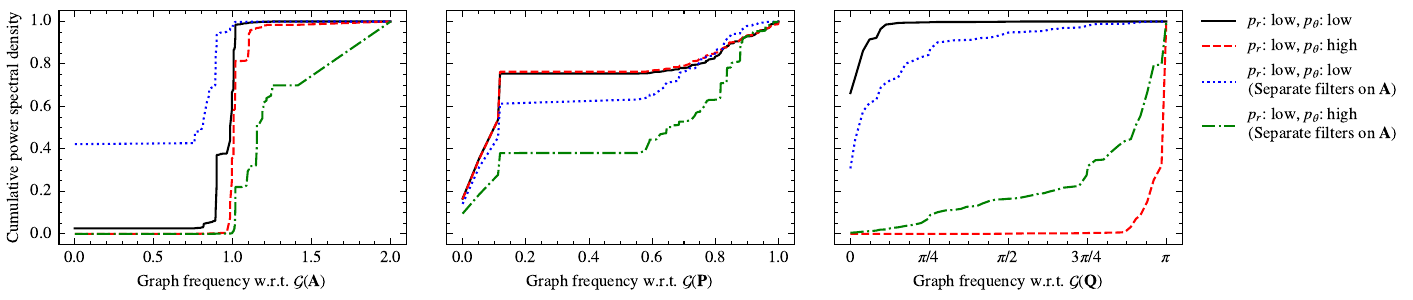}
	\caption{
 The filtered signals reveal different spectra for different variation definitions on \(\adjmat\), \(\polarPsdLeft\), and \(\polarUni\). 
 While all the filtered signals localized around mid frequency on \(\adjmat\), they appear completely different on \( \polarPsdLeft \) and \( \polarUni \) corresponding to each filter type.
 }
	\label{fig:diffusion_exp}
\end{figure*} 

\subsubsection*{Filter design}
For the normal case, there is a direct correspondence of the cascaded polynomial filtering \(p_\theta(\polarUni)p_r(\polarPsdLeft)\) in terms of polynomial filters on the magnitude and phase of eigenvalues of \(\adjmat\). 
\begin{theorem}
Let \(\adjmat\in\smat{n}{R}\) be a real-valued normal matrix with \(\adjmat\adjmat\tran = \adjmat\tran\adjmat\) and eigendecomposition \( \adjmat=\eigmat\eigvalmat\eigmat\herm \), where \(\eigmat\in\suni{n}{C}\) and \(\eigvalmat\in\sdiag{n}{C}\).  When \(\sqrt{{\eigvalmat}\eigvalmat\herm}=\mat{\Sigma}\in\sdiag{n}{R}\cap\spsd{n}{R}\) and \(\eigvalmat=\mat{\Sigma}\mat{\Theta}\), 
for any two polynomial pairs \(p_\theta\) and \(p_r\),
\begin{equation}
\label{eq:PQ-polynomials}p_\theta(\polarUni)p_r(\polarPsdLeft)=\eigmat p_\theta(\mat{\Theta})p_r(\mat{\Sigma})\eigmat\herm,
\end{equation}
where \(p_\theta(\mat{\Theta})p_r(\mat{\Sigma})\) are ``separate'' filters on \(\adjmat\), with \(p_\theta\) targeting the phase and \(p_r\) addressing the magnitude of the eigenvalues.
\end{theorem}
\begin{proof}
The eigendecomposition of \(\adjmat=\eigmat\eigvalmat\eigmat\herm=\eigmat\mat{\Sigma}\mat{\Theta}\eigmat\herm=\eigmat\mat{\Sigma}(\eigmat\mat{\Theta}\herm)\herm=\eigmat\mat{\Sigma}\svdRight\herm\) where \( \svdRight=\eigmat\mat{\Theta}\herm \). Since the multiplication of two unitary matrices produces another unitary matrix, \(\svdRight\) is a unitary matrix. Therefore, \(\adjmat=\eigmat\mat{\Sigma}\svdRight\herm\) is a singular value decomposition (SVD) of the adjacency matrix.

By the relationship between SVD and the polar decomposition, 
\begin{align}
    \polarUni &= \eigmat\svdRight\herm = \eigmat(\eigmat\mat{\Theta}\herm)\herm= \eigmat\mat{\Theta}\eigmat\herm,\\
    \polarPsdLeft &= \svdRight\mat{\Sigma}\svdRight\herm =(\eigmat\mat{\Theta}\herm)\mat{\Sigma}(\eigmat\mat{\Theta}\herm)\herm=\eigmat\mat{\Sigma}\eigmat\herm.
\end{align}
Therefore, 
\begin{align}
p_\theta(\polarUni)p_r(\polarPsdLeft)&=p_\theta(\eigmat\mat{\Theta}\eigmat\herm)p_r(\eigmat\mat{\Sigma}\eigmat\herm)\\
&=\eigmat p_\theta(\mat{\Theta})\eigmat\herm\eigmat p_r(\mat{\Sigma})\eigmat\herm\\
&=\eigmat p_\theta(\mat{\Theta}) p_r(\mat{\Sigma})\eigmat\herm.\qedhere
\end{align}
\end{proof}
A crucial implication of the presented theorem is the feasibility of implementing separate filters targeting the magnitude and phase of the eigenvalues of \( \adjmat \). This is particularly advantageous for large graphs as it obviates the need for the computationally intensive eigendecomposition of \( \adjmat \). Instead, we leverage cascaded filtering with \( \polarUni \) and \( \polarPsdLeft \). In subsequent sections, our experimental results demonstrate that this cascaded filtering approach can also be applied in cases where the adjacency matrix deviates from normality.

\section{Case study: M-block cyclic graph}
\label{sec:case_study}

This section aims to highlight that even in the context of non-normal matrices, where the direct correspondence of \eqref{eq:PQ-polynomials} does not hold, there still exists a significant advantage in analyzing graph signals using \( \graph{\polarPsdLeft} \), \( \graph{\polarPsdRight} \), and \( \graph{\polarUni} \) individually. 

We select an \( M \)-block cyclic graph for our case study~\cite{mblock}.
An \( M \)-block cyclic graph comprises \( M \) distinct blocks of nodes, where edges interlink nodes in adjacent blocks, but there are no connections within the same block. All edge orientations follow a uniform cyclic pattern, either clockwise or counter-clockwise. Fig.~\ref{fig:symm_examples} depicts an \( M \)-block cyclic graph divided into \( M=4 \) blocks. Within this structure, indirect variation is defined by changes between nodes in the same block, while in-flow variation relates to changes between nodes in neighboring blocks. Notably, for balanced \( M \)-block cyclic graphs, where each block contains an equal number of nodes, \( \graph{\polarUni} \) retains the \( M \)-block cyclic layout of the original \( \graph{\adjmat} \) but with altered edge weights.

To investigate signal variations across different variations on \( \graph{\adjmat} \), \( \graph{\polarPsdLeft} \), and \( \graph{\polarUni} \) (omitting \( \graph{\polarPsdRight} \) for conciseness), we design graph signals with distinct smoothness levels as follows.
We begin by generating an independent and identically distributed (iid) graph signal \( \gsignalVec \), with its entries drawn from the standard normal distribution \( \gaussian{0}{1} \). Subsequently, we employ cascaded filters on this random signal and study the resultant filtered signal, \(\vec{y} = p_\theta(\polarUni)p_r(\polarPsdLeft)\gsignalVec\), through diverse variations. For each filter, we selected varying types (lowpass and highpass), and some resultant filtered outputs are illustrated in 
Fig.~\ref{fig:filter_output}. 
This analysis brings forth several insights:
1) The lowpass filter \(p_r\) smoothens signals within a block, reflecting the indirect variation.
2) Filters denoted by \(p_\theta\) influence the smoothness from one block to another in the signal.
3) Applying separate filters on the magnitude and phase of the eigenvalue of \(\adjmat\) produces qualitatively similar signals to those processed by the cascaded filters using identical \(p_r\) and \(p_\theta\) although the equality \eqref{eq:PQ-polynomials} does not hold for non-normal \(\adjmat\).

Figure~\ref{fig:diffusion_exp} illustrates the spectrum of each filtered signal across various variations. For signals measured on \( \adjmat \), all the filtered signals are localized around the mid frequency. However, for \( \graph{\polarPsdLeft} \) and \( \graph{\polarUni} \), their graph frequency spectra localize according to the corresponding filter types. This characteristic is consistent with the understanding that signal variations on \( \graph{\polarPsdLeft} \) are measured between nodes within the same block in the context of \( M \)-block cyclic graphs.
Conversely, due to the inherent Hamiltonian cycles in all sub-graphs of \( \graph{\polarUni} \) and its structure resonating with the \( M \)-block cyclic nature of \( \graph{\adjmat} \), \( \graph{\polarUni} \) features cycles with lengths that are integer multiples of \( M=4 \). As a result, signal variations on \( \graph{\polarUni} \) are measured between nodes in successive blocks, consistent with our earlier observations. It is worth noting that such intra-block connections might not necessarily be present in \( \graph{\adjmat} \). Also, signals filtered with lowpass and highpass filters denoted by \( p_\theta \), which respectively admit angular frequencies approximately \( 0 \) and \( \pi \) (radians per node), correspond to zero and one block-to-block cycles, respectively.

\section{Conclusions}
\label{sec:conclusions}
In this study, we delve into the analysis of signals on directed graphs leveraging polar decomposition. With a clear node-domain interpretation, we elucidate that the positive-semidefinite matrix offers insights into the indirect variations present within the original graph, whereas the unitary matrix underscores in-flow variations. Furthermore, we propose cascaded graph filters based on the polar decomposition and illustrate their equivalency to separable filtering on the magnitude and phase of the eigenvalues of \( \adjmat \). Through a focused case study on an M-block cyclic graph, we demonstrate the efficacy of cascaded filters in smoothing graph signals in terms of indirect and in-flow variations distinctly. Our findings highlight that analyzing these filtered signals in light of the aforementioned variations furnishes a richer perspective than solely probing signal variations from the original graph.

\newpage

\bibliographystyle{IEEEbib}
\bibliography{refs}

\end{document}